%% file: master.tex
\documentclass[a4paper,UKenglish,cleveref, autoref,numberwithinsect]{lipics-v2019}

\nolinenumbers

\usepackage{latexsym}
\usepackage{amssymb}
\usepackage{amsmath}
\usepackage{amsthm}
\usepackage{thmtools, thm-restate}
\usepackage{stmaryrd}
\usepackage{amsfonts}
\usepackage{xspace}
\include{macros}
\title{Invariants for Continuous Linear Dynamical Systems} 

\author{Shaull Almagor}{Computer Science Department, Technion, Israel }{shaull@cs.technion.ac.il}{https://orcid.org/0000-0001-9021-1175}{Supported by a European Union's Horizon 2020 research and innovation programme under the Marie Sk{\l}odowska-Curie grant agreement No 837327.}

\author{Edon Kelmendi}{Department of Computer Science, Oxford University, UK}{edon.kelmendi@cs.ox.ac.uk}{}{}

\author{Jo\"el Ouaknine}{Max Planck Institute for Software Systems,  
	Germany \and Department of Computer Science, Oxford
        University, UK}{joel@mpi-sws.org}{}{Supported by ERC grant
        AVS-ISS (648701) and DFG grant 389792660 as part of TRR
248 (see \href{https://perspicuous-computing.science}{https://perspicuous-computing.science}).}

\author{James Worrell}{Department of Computer Science, Oxford University, UK}{jbw@cs.ox.ac.uk}{}{Supported by EPSRC Fellowship EP/N008197/1.}

\authorrunning{S. Almagor, E. Kelmendi, J. Ouaknine, and J.Worrell}

\Copyright{Shaull Almagor, Edon Kelmendi, Jo\"el Ouaknine, and James Worrell}



\category{}

\relatedversion{}

\supplement{}

\EventEditors{John Q. Open and Joan R. Access}
\EventNoEds{2}
\EventLongTitle{42nd Conference on Very Important Topics (CVIT 2016)}
\EventShortTitle{CVIT 2016}
\EventAcronym{CVIT}
\EventYear{2016}
\EventDate{December 24--27, 2016}
\EventLocation{Little Whinging, United Kingdom}
\EventLogo{}
\SeriesVolume{42}
\ArticleNo{23}

\begin{document}
\maketitle
\begin{abstract}
  Continuous linear dynamical systems are used extensively in
  mathematics, computer science, physics, and engineering to model the
  evolution of a system over time. A central technique for certifying
  safety properties of such systems is by synthesising inductive invariants.
  This is the task of finding a set of states that is closed under the
  dynamics of the system and is disjoint from a given set of error
  states.  In this paper we study the problem of synthesising
  inductive invariants that are definable in o-minimal expansions of
  the ordered field of real numbers.  In particular, assuming
  Schanuel's conjecture in transcendental number theory, we establish
  effective synthesis of o-minimal invariants in the case of
  semi-algebraic error sets.  Without using Schanuel's conjecture, we
  give a procedure for synthesizing o-minimal invariants that
  contain all but a bounded initial segment of the orbit and are
  disjoint from a given semi-algebraic error set.  We further prove
  that effective synthesis of semi-algebraic invariants
  that contain the whole orbit, is at least as hard as a certain open
  problem in transcendental number theory.
\end{abstract}
\pagebreak
\input{introduction}
\input{prelims}
\input{cones}
\input{fatcones}

\input{hardness}

\bibliography{bibliography}

\appendix
\input{proofProp}
\input{containscone}
\input{proofsSec4}
\input{proofsSec5}
\end{document}

%% file: macros.tex
  {\par\addvspace{.6pc plus .2pc minus .1pc}
   \noindent
   {\bf\proofname}\enspace\ignorespaces}%
  {\par\addvspace{.6pc plus .2pc minus .1pc}}
\def\proofname{Proof.}
\def\qed{\relax\ifmmode\hfill \Box\else\unskip\nobreak\hfill $\Box$\fi}

\newcommand{\defequals}{\stackrel{\mathrm{def}}{=}}

\newcommand{\set}[1]{\{#1\}}
\newcommand{\tup}[1]{\langle #1 \rangle}
\renewcommand{\st}{\ :\ }
\newcommand{\orb}{\mathcal{O}}
\newcommand{\inv}{\mathcal{I}}
\newcommand{\diag}{\mathrm{diag}}
\newcommand{\torus}{\mathbb{T}}
\newcommand{\cone}{\mathcal{C}}
\newcommand{\ray}{\mathrm{r}}

\newcommand{\newextmathcommand}[2]{%
	\newcommand{#1}{\ensuremath{#2}\xspace}
}

\newextmathcommand{\theo}{\mathfrak{R}}
\newextmathcommand{\theosmaller}{\mathfrak{R}_{1}}
\newextmathcommand{\theoreals}{\mathfrak{R}_{0}}
\newextmathcommand{\theoexp}{\mathfrak{R}_{\exp}}
\newextmathcommand{\theorest}{\mathfrak{R}^{\rm RE}}
\newextmathcommand{\theoexprest}{\mathfrak{R}^{\rm RE}_{\exp}}
\newextmathcommand{\theopow}{\mathfrak{R}_{\rm pow}}

\newcommand{\QQ}{\mathbb Q}
\newcommand{\RR}{\mathbb R}
\newcommand{\CC}{\mathbb C}
\newcommand{\ZZ}{\mathbb Z}
\newcommand{\NN}{\mathbb N}

\newcommand{\cV}{{\cal V}}
\newcommand{\cF}{{\cal F}}

\newcommand{\bbS}{\mathbb{S}}

\newcommand{\vect}[1]{\boldsymbol{\mathbf{#1}}}

\newcommand{\ii}{\mathrm{i}}

\newcommand{\norm}[1]{{\left\lVert#1\right\rVert}}
\newcommand{\lowervec}{\boldsymbol{\ell}}
\newcommand{\uppervec}{\mathbf{u}}
\newcommand{\boxvec}{{\rm Box}}


%% file: introduction.tex
\section{Introduction}
\label{sec:intro}
A \emph{continuous linear dynamical system} (CDS) is a system whose
evolution is governed by a differential equation of the form
$\dot{\vect x}(t)=A\vect x(t)$, where $A$ is a matrix with real
entries. CDSs are ubiquitous in mathematics, physics, and engineering;
they have been extensively studied as they describe the evolution of
many types of systems (or abstractions thereof) over time. More
recently, CDSs have become central in the study of cyber-physical
systems (see, e.g., the textbook~\cite{alur2015principles}).

In the study of CDSs, particularly from the perspective of control
theory, a fundamental problem is \emph{reachability}---namely whether
the orbit $\{\vect x(t)\st t\ge 0\}$ intersects a given target set
$Y\subseteq \RR^d$. For example, when $\vect x(t)$ describes the state
of an autonomous car (i.e., its location, velocity, etc.). $Y$ may
describe situations where the car is not able to stop in time to
respond to a hazard.

When $Y$ is a singleton set, reachability is decidable~\cite[Theorem
  2]{hainry}.  However, already when $Y$ is a half-space it is open
whether or not reachability is decidable.  The latter decision problem
is known in the literature as the \emph{continuous Skolem
  problem}. Some partial positive results were given
in~\cite{Bell10}~and~\cite{chonev2016skolem}. The continuous Skolem
problem is related to notoriously difficult problems in the theory of
Diophantine approximation: specifically a procedure for the continuous
Skolem problem would yield one for computing to arbitrary precision
the \emph{Diophantine-approximation types} of all real algebraic
numbers~\cite{chonev2016skolem}.

In lieu of an algorithm to decide reachability, one approach is to
find a set $X$ that separates the orbit from $Y$.  In order for this
scheme to be useful, structural restrictions are placed on $X$ to make
it easy to verify that $X$ contains the orbit and that it is disjoint
from $Y$ (indeed, if we give up either requirement, we can use as $X$
either the orbit itself, or $\RR^d\setminus Y$, neither of which makes
the problem any easier).

Natural candidates for such structured sets are \emph{inductive
  invariants}. These are sets that are invariant under the
dynamics of the system. If $X$ is an inductive invariant, proving that
the orbit is contained in $X$ amounts to proving that the starting
point $\vect x(0)$ belongs to $X$, which is typically easy. Further by
restricting the class of sets under consideration (e.g., polyhedra,
semi-algebraic sets, etc.), testing whether $X$ intersects $Y$
becomes, likewise, easy.

The papers~\cite{ACOW20, ACOW18} study o-minimal invariants for
\emph{discrete} linear dynamical systems. There it is proved that when
the target $Y$ is a semi-algebraic set, the question of whether there
exists an o-minimal invariant disjoint from $Y$ is decidable. Furthermore,
if there is an o-minimal invariant then there is in fact a
semi-algebraic invariant which can moreover be constructed
effectively. The present paper uses similar ideas, although the case of continuous
linear dynamical systems differs in several important ways.

\subparagraph*{Main Contributions.}
We consider the following problem: given a CDS by means of a matrix
$A$ with rational entries, an initial point $\vect x_0=\vect x(0)$,
and a semi-algebraic set $Y$ of error states, decide whether there
exists a set that is definable in some o-minimal expansion of the
ordered real field and is (1) disjoint from $Y$, (2) invariant under
the dynamics of the system, and (3) contains the initial point $\vect
x_0$.  We show that in searching for such invariants it suffices to
look among sets definable in the expansion of the reals with the real
exponential function and trigonometric functions restricted to bounded
domains.  Moroever, assuming Schanuel's conjecture (a unifying
conjecture in transcendental number theory), we prove that the
existence of such an invariant is decidable, and that invariants
can effectively be constructed when they exist.

Without assuming Schanuel's conjecture we can decide a related
problem, namely the question of whether there exists a set that is
definable in an o-minimal expansion of the real field and is (1)
disjoint from $Y$, (2) invariant under the dynamics of the system, and
(3) meets the orbit of the initial point $\vect x_0$.  Notice that
such a set---which could be called an \emph{eventual invariant}---must
contain all but a bounded initial segment of the orbit.  We show that
when such a set exists, it can be effectively constructed and moreover
that it can be chosen to be a semi-algebraic set.  Such an invariant
can serve as a certificate that the orbit does not enter the error set
$Y$ infinitely often. The latter is a very difficult problem to
decide, even when the target set is a half-space
\cite{chonev2016recurrent}.

As mentioned earlier, for discrete linear dynamical systems the
question of whether there exists a semi-algebraic invariant that
contains the {\em whole} orbit is decidable~\cite{ACOW20, ACOW18}. We
provide an explanation of why the analogous result for continuous
systems is not easy to prove; this is by way of a reduction from a difficult
problem that highlights the complications of continuous systems. The
problem asks whether a given exponential polynomial of the form
\begin{align*}
f(t)=a_1e^{b_1t}+\cdots+a_ne^{b_nt}
\end{align*}
has zeros in a bounded interval, where $a_i,b_i$ are real algebraic
numbers. Deciding whether $f$ has zeros in a bounded region seems to
be difficult because all the zeros have to be transcendental (a
consequence of Hermite-Lindemann Theorem), and they can be tangential,
i.e., $f$ never changes its sign, yet it has a zero.

\subparagraph*{Related Work.}  Invariant synthesis is a central technique for establishing safety properties of hybrid systems. It has long been known how to compute a strongest \emph{algebraic} invariant~\cite{Rodriguez-CarbonellT05} (i.e., a smallest algebraic set that contains the collection of reachable states) for an arbitrary CDS.  Here an algebraic invariant is one that is specified by a conjunction of polynomial equalities.  If one moves to the more expressive setting of semi-algebraic invariants, which allow inequalities, then there is typically no longer a strongest (or smallest) invariant, but one can still ask to decide the existence of an invariant that avoids a given target set of configurations. This is the problem that is addressed in the present paper. 

Partial positive results are known, for example when strong
restrictions on the matrix $A$ are imposed, such as when all the
eigenvalues are real and rational, or purely imaginary with rational
imaginary part~\cite{LafferrierePY01}.

A popular approach in previous work has been to seek invariants that match a given syntactic \emph{template}, which allows to reduce invariant synthesis to constraint solving~\cite{GulwaniT08,SturmT11,LiuZZ11}.  While this technique can be applied to much richer classes of systems than those considered here (e.g., with discrete control modes and non-linear differential equations), it does not appear to offer a way to decide the existence of arbitrary semi-algebraic invariants.  An alternative to the template approach for invariant generation involves obtaining candidate invariants from semi-algebraic abstractions of a system~\cite{SogokonGJP16}.  Another active area of current research lies in developing powerful techniques to check whether a given semi-algebraic set is actually an invariant~\cite{GhorbalSP17,LiuZZ11}.

Other avenues for analysing dynamical systems in the literature
include bisimulations~\cite{broucke2002reachability}, forward/backward
reach-set computation~\cite{anai2001reach}, and methods for directly
proving liveness properties~\cite{sogokon2015direct}. The latter
depends on constructing {\em staging sets}, which are essentially
semi-algebraic invariants.

Often, questions about dynamical systems can be reduced to deciding
whether a sentence belongs to the elementary theory of an appropriate
expansion of the ordered field of real numbers. While the latter is
typically undecidable, there are partial positive results, namely 
quasi-decidability in bounded domains, see~\cite{franek2016quasi} and
the references therein. This can be used to reason about the dynamics
of a system in a bounded time interval, under the assumption that it
does not tangentially approach the set that we want to avoid. However,
it seems unlikely that such results can be easily applied to the
problems considered here.

The rest of the paper is organised as follows. In \cref{sec:prelim},
we give the necessary definitions and terminology. In \cref{sec: orbit
  cones}, we define {\em cones}, which are over-approximations of the
orbit, and prove that they are in a certain sense canonical. The
positive results assuming Schanuel's conjecture are subsequently given
in this section. \cref{sec: fat} is devoted to the effective
construction of the semi-algebraic invariants which allows us to state
and prove the unconditional positive results. In \cref{sec:hardness},
we give the aforementioned reduction, from finding zeros of
exponential polynomials.



%% file: prelims.tex
\section{Preliminaries}
\label{sec:prelim}

A \emph{continuous-time linear dynamical system} is a pair
\begin{align*}
  \tup{A,\vect x_0}
\end{align*}
where $A \in \mathbb{Q}^{d \times d}$ and $\vect x_0 \in \mathbb{Q}^d$. The system evolves in time according the function $x(t)$ which is the unique solution to the differential equation $\dot{\vect x}(t)=A\vect x(t)$ with $\vect x(0)=\vect x_0$. Explicitly this solution can be written as:
\begin{align*}
 \vect x(t)=e^{At}\vect x_0. 
\end{align*}

The \emph{orbit of $\tup{A,\vect x_0}$ from time $t_0$} is the set $\orb(t_0)=\set{e^{At}\vect x_0 \st t\ge t_0}$.  An \emph{invariant for $\tup{A,\vect x_0}$ from time $t_0$} is a set $\inv \subseteq \mathbb{R}^d$ that contains $e^{At_0}\vect x_0$ and is stable under applications of $e^{At}$, i.e., $e^{At} \inv \subseteq \inv$ for every $t\ge 0$. Note that an invariant from time $t_0$ contains $\orb(t_0)$.  Given a set $Y \subseteq \RR^d$ (referred to henceforth as an \emph{error set}), we say that the invariant $\inv$ \emph{avoids} $Y$ if the two sets are disjoint.  

We denote by $\theoreals$ the structure $\tup{\RR,0,1,+,\cdot,<}$. This is the ordered field of real numbers with constants $0$ and $1$.  A sentence in the corresponding first-order language is a quantified Boolean combination of atomic propositions of the form $P(x_1,\ldots,x_n)>0$, where $P$ is a polynomial with integer coefficients and $x_1,\ldots,x_n$ are variables. 
In addition to $\theoreals$, we also consider its following expansions:
\begin{itemize}
	\item $\theoexp$, obtained by expanding $\theoreals$ with the real exponentiation
	function $x \mapsto e^x$.
	\item $\theorest$, obtained by expanding $\theoreals$ with the \emph{restricted elementary functions}, namely $x\mapsto e^x|_{[0,1]}$, $x\mapsto \sin x|_{[0,1]}$, and $x\mapsto \cos x|_{[0,1]}$.
	\item $\theoexprest$, obtained by expanding $\theoexp$ with the restricted elementary functions.
\end{itemize}

Tarski famously showed that the first-order theory of $\theoreals$ admits quantifier elimination, moreover the elimination is effective and therefore the theory is decidable~\cite[Theorem 37]{tarski1951decision}. 

It is an open question whether the theory of the reals with exponentiation (\theoexp) is decidable; however decidability was established subject to Schanuel's conjecture by MacIntyre and Wilkie~\cite[Theorem 1.1]{MacintyreWilkie1996}.  MacIntyre and Wilkie further showed in~\cite[Section 5]{MacintyreWilkie1996} that decidability of the theory of $\theoexp$ implies a weak form of Schanuel's conjecture.

Similarly, it is an open question whether \theorest and \theoexprest are decidable, but they are also known to be decidable subject to Schanuel's conjecture~\cite[Theorem 3.1]{macintyre2016turing}\footnote{More precisely, the decidability of \theoexp requires Schanuel's conjecture over $\RR$, whereas that of \theoexprest requires it over $\CC$.}.

Let $\theo$ be an expansion of the structure $\theoreals$.
A set $S\subseteq \RR^d$ is \emph{definable} in $\theo$ if
there exists a formula $\phi(x_1,\ldots,x_d)$ in $\theo$ with free
variables $x_1,\ldots,x_d$ such that $S=\set{(c_1,\ldots,c_d)\in \RR^d
	\mid \theo \models \phi(c_1,\ldots,c_d)}$. 
For $\theo = \theoreals$, the ordered field of real numbers,
$\theoreals$-definable sets 
are known as
\emph{semi-algebraic} sets.


\begin{remark}
  \label{rmk:complex definable}
  There is a natural first-order interpretation of the field of complex numbers $\CC$ in the field of real
  numbers $\RR$.  We shall say that a set $S\subseteq \CC^d$ is \emph{$\theo$-definable} if the image
  $\{ (x, y) \in \RR^d \times \RR^d \mid x + i y \in S \}$ of $S$ under this interpretation is $\theo$-definable.
\end{remark}

A totally ordered structure $\tup{M,<,\ldots}$ is said to be \emph{o-minimal} if every definable subset of $M$ is a finite union of intervals.  Tarski's result on quantifier elimination implies that $\theoreals$ is o-minimal. The o-minimality of $\theoexp$ and $\theorest$ is shown in~\cite{Wilkie96}, and the o-minimality of \theorest and \theoexprest is due to~\cite{van1994elementary,van1996geometric}.

 A \emph{semi-algebraic invariant} is one that is definable in $\theoreals$. An \emph{o-minimal invariant} is one that is definable in an o-minimal expansion of $\theoexp$.


%% file: cones.tex
\section{Orbit Cones}
\label{sec: orbit cones}

In this section we define orbit cones, an object that plays a central role in the subsequent results. They can be thought of as over-approximations of the orbit that has certain desirable properties, and moreover it is canonical in the sense that any other invariant must contain a cone.

\subsection{Jordan Normal Form}
Let $\tup{A,\vect x_0}$ be a continuous linear dynamical system. The exponential of a square matrix $A$ is defined by its formal power series as
\begin{align*}
  e^A\defequals\sum_{n=0}^\infty \frac{A^n}{n!}.
\end{align*}
Let $\lambda_1,\ldots,\lambda_k$ be the eigenvalues of $A$, and recall that when $A\in \QQ^{d\times d}$, all the eigenvalues are algebraic. We can write $A$ in \emph{Jordan Normal Form} as $A=PJP^{-1}$ where $P\in \CC^{d\times d}$ is an invertible matrix with algebraic entries, and $J=\diag(B_1,\ldots, B_k)$ is a block-diagonal matrix where each block $B_l$ is a Jordan block that corresponds to eigenvalue $\lambda_l$, and it has the form
\[
B_l = \begin{pmatrix}
\lambda_l & 1 & 0 & \cdots & 0 \\
0 & \lambda_l & 1 & \cdots & 0 \\
\vdots & \vdots & \vdots & \vdots & \vdots\\
0 & 0 & 0 & \cdots & \lambda_l
\end{pmatrix}\in \CC^{d_l\times d_l}
\]
with $\sum_{l=1}^k d_l=d$.


From the power series, we can write $e^{At}=Pe^{Jt}P^{-1}$. Further, $e^{Jt}=\diag(e^{B_1},\ldots, e^{B_k})$.  For each $1\le l\le k$, write $B_l=\Lambda_l+N_l$, where $\Lambda_l$ is the $d_l\times d_l$ diagonal matrix $\diag(\lambda_l,\ldots,\lambda_l)$ and $N_l$ is the $d_l\times d_l$ matrix $\diag_2 (1,\ldots ,1)$; where $\diag_j(\cdot)$ is the $j$-th diagonal matrix, with other entries zero. 
	
The matrices $\Lambda_l$ and $N_l$ commute, since the former is a diagonal matrix. A fundamental property of matrix exponentiation is that if matrices $A,B$ commute, then $e^{A+B}=e^Ae^B$.  Thus, we have
\begin{align*}
  e^{Jt}=e^{\diag(\Lambda_1 t+N_1t, \ldots, \Lambda_k t+N_kt)}=\diag(e^{\lambda_1t},\ldots,e^{\lambda_kt}) e^{\diag(N_1t,\ldots,N_kt)},
\end{align*}
where by $\diag(e^{\lambda_1t},\ldots,e^{\lambda_kt})$ we mean the $d\times d$ diagonal matrix that has the entry $e^{\lambda_1t}$ written $d_1$ times, the entry $e^{\lambda_2t}$ written $d_2$ times and so on. It will always be clear from the context whether we repeat the entries because of their multiplicity or not. 
	
Matrices $N_l$ are nilpotent, so its power series expansion is a finite sum, i.e. a polynomial in $N_l t$. More precisely, one can verify that:
\begin{align*}
e^{N_lt}=I+\diag_2(t,\ldots,t)+\diag_3(\frac{t^2}{2},\ldots,\frac{t^2}{2})+\ldots+\diag_{d_l}
\left(\frac{t^{(d_l-1)}}{(d_l-1)!}
\right).
\end{align*}
Set $Q(t)\defequals\diag(e^{N_1 t},\ldots, e^{N_k t})$. From the equation above, the entries of $Q(t)$ are polynomials in $t$ with rational coefficients.

Write the eigenvalues as $\lambda_l=\rho_l+\ii\omega_l$, so that
\begin{align*}
  \diag(e^{\lambda_1t},\ldots,e^{\lambda_kt})=\underbrace{\diag(e^{\rho_1 t},\ldots,e^{\rho_k t})}_{E(t)}\cdot \underbrace{\diag(e^{\omega_1\ii t},\ldots,e^{\omega_k\ii t})}_{R(t)}
\end{align*}


We have in this manner decomposed the orbit
\begin{align*}
  \orb(t_0)=\set{P\ \ E(t)\ R(t)\ Q(t)\ \ P^{-1}\vect x_0\st t\ge t_0},
\end{align*}
into an exponential $E(t)$, a rotation $R(t)$, and a simple polynomial $Q(t)$ matrices that commute with one another. Having the orbit in such a form will facilitate the analysis done in the sequel. 

\subsection{Cones as Canonical Invariants}      

In a certain sense, the rotation matrix $R(t)$ is the most complicated, because of it, the orbit is not even definable in $\theoexp$. The purpose of cones is to abstract away this matrix by a much simpler subgroup of the complex torus
\begin{align*}
  \torus \defequals \set{\vect z \in \CC^k \st |z_i|=1, 1\le i\le k}.
\end{align*}

To this end, consider the group of additive relations among the frequencies $\omega_1,\ldots, \omega_k$:
\begin{align*}
  S \defequals \set{\vect a\in \ZZ^k \st a_1\omega_1+\cdots +a_k\omega_k=0}.
\end{align*}
The subgroup of the torus of interest, respects the additive relations as follows:
\begin{align*}
  \torus_\omega \defequals \set{(\tau_1,\ldots,\tau_k)\in\torus\st \text{for all }\vect a\in S,\ \tau_1^{a_1}\cdots\tau_k^{a_k}=1}. 
\end{align*}
Its desirable properties are summarised in the following proposition:
\begin{proposition}
  \label{prop:dense subgroup}
  For algebraic numbers $\omega_1,\ldots,\omega_k$, 
  \begin{enumerate}
  \item $\torus_\omega$ is semi-algebraic,
  \item diagonals of $\set{R(t)\st t\ge 0}$ form a dense subset of $\torus_\omega$. 
  \end{enumerate}
\end{proposition}
\begin{proof}
  Being an Abelian subgroup of $\ZZ^k$, $S$ has a finite basis, moreover this basis can be computed because of effective bounds, \cite[Section 3]{Mas88}. To check that $(\tau_1,\ldots,\tau_k)$ belongs to $\torus_\omega$, it suffices to check that $\tau_1^{a_1}\cdots\tau_k^{a_k}=1$ for $(a_1,\ldots,a_k)$ in the finite basis. This forms a finite number of equations, therefore $\torus_\omega$ is semi-algebraic. The fact that this is a subset of vectors of complex numbers is not problematic in this case because of the simple first-order interpretation in the theory of reals, see \cref{rmk:complex definable}.

  The second statement of the proposition is a consequence of Kronecker's theorem on inhomogeneous simultaneous Diophantine approximations, see \cite[Page 53, Theorem 4]{cassels1965introduction}. The proof of a slightly stronger statement can also be found in \cite[Lemma 4]{chonev2016recurrent}. Examples can be found where the set of diagonals of $\set{R(t)\st t\ge 0}$ is a strict subset of $\torus_\omega$. 
\end{proof}

The orbit cone can now be defined by replacing the rotations with the subgroup of the torus. As it turns out, for our purposes this approximation is not too rough. 

\begin{definition}
  \label{def:orbit cone}
  The \emph{orbit cone} from $t_0\geq 0$ is
  \begin{align*}
    \cone_{t_0} \defequals \left\{P\ E(t)\ \diag(\vect \tau)\ Q(t)\ P^{-1}\vect x_0\st \vect \tau\in\torus_\omega,t\geq t_0 \right\}.
  \end{align*}
\end{definition}

We prove that the cone is an inductive invariant and also a subset of $\RR^d$. 

\begin{lemma}
  \label{lem: cone is invariant}
  For all $\delta,t_0\geq 0$, $e^{A\delta}\cone_{t_0}\subseteq \cone_{t_0}$.
\end{lemma}
\begin{proof}
  Fix $t\geq t_0$ and $\vect \tau\in\torus_\omega$, and consider the point
  \begin{align*}
    \vect v=P\ E(t)\ \diag(\vect \tau)\ Q(t)\ P^{-1}\vect x_0\in \cone_{t_0},
  \end{align*}
  then we can write $e^{A\delta}\vect v$ as
  \begin{align*}
    e^{A\delta}\vect v&= P\ E(\delta)R(\delta)Q(\delta)\cdot E(t)\diag(\vect \tau)Q(t)\ P^{-1}\vect x_0\\
    &=P\ E(\delta+t)\ R(\delta)\diag(\vect \tau)\ Q(\delta)Q(t)\ P^{-1}\vect x_0. 
  \end{align*}
  The matrix $R(\delta)\diag(\vect \tau)$ is equal to $\diag(\vect \tau')$ for some $\vect \tau'\in\torus_\omega$. Otherwise said, the vector $(e^{\delta\omega_1\ii}\tau_1,\ldots,e^{\delta\omega_k\ii}\tau_k)$ belongs to $\torus_\omega$. Indeed this is the case because for any $\vect a\in S$ we have
  \begin{align*}
    e^{a_1\delta\omega_1\ii}\tau_1^{a_1}\cdots e^{a_k\delta\omega_k\ii}\tau_k^{a_k}=e^{\delta\ii\ (a_1\omega_1+\cdots +a_k\omega_k)}\cdot\tau_1^{a_1}\cdots \tau_k^{a_k}=1. 
  \end{align*}
  Finally, by induction on the dimension $d$ one can verify that $Q(\delta)Q(t)=Q(\delta+t)$. 
\end{proof}

The fact that cones are subsets of $\RR^d$ comes as  a corollary of the following proposition which is proved in \cref{apx: cone is real}.
\begin{restatable}{proposition}{propreals}
  \label{prop: real entries}
  Let $A=PJP^{-1}$ as above, and let $C_i\in \CC^{d_i\times d_i}$ for $i=1,\ldots, k$, with dimensions compatible to the Jordan blocks of $A$, and such that for every $i_1,i_2$, if $B_{i_1}=\overline{B_{i_2}}$, then $C_{i_1}=\overline{C_{i_2}}$. Then $P\diag(C_1,\ldots,C_k)P^{-1}$ has real entries.
\end{restatable}

The matrix $E(t)\diag(\tau)Q(t)$ can be written as $\diag(C_1,\ldots,C_k)$ where the $C_i$ matrices satisfy the
conditions of \cref{prop: real entries}, hence the following corollary. 
\begin{corollary}
  \label{cor: cone is real}
  For all $t_0\ge 0$ we have  $\cone_{t_0}\subseteq \RR^d$.

\end{corollary}
It is surprising that, already, the cones are a complete characterisation of o-minimal inductive invariants in the following sense. 
\begin{restatable}{theorem}{eventualcone}
  \label{thm: eventual cone contained in o-min invariant}
  Let $\inv$ be an o-minimal invariant that contains the orbit $\orb(u)$ from some time $u\ge 0$, then there exists $t_0\ge u$ such that:
  \begin{align*}
    \cone_{t_0}\subseteq \inv.
  \end{align*}
\end{restatable}
\begin{proof}[Proof sketch]
  Conceptually, the proof follows along the lines of its analogue in~\cite{ACOW18}. There are a few differences, namely that the entries of the matrix $A$ in~\cite{ACOW18} are assumed to be algebraic, while this is not true for the entries of $e^A$.

  We define rays of the cone, which are subsets where $\vect \tau\in\torus_\omega$ is fixed. Then we prove that for every ray, all but a finite part of it, is contained in the invariant. This is done by contradiction: if a ray is not contained in the invariant, a whole dense subset of the cone can be shown not to be contained in the invariant, leading to a contradiction, since the invariant is assumed to contain the orbit. We achieve this using some results on the topology of o-minimal sets.

  The complete proof deferred to \cref{apx: invariant contains cone}.
	
\end{proof}

Another desirable property of cones is that they are \theoexp-definable. Also, one can observe that for every $t_0$, the set $\set{e^{At}\vect x_0: 0\le t\le t_0}$ is definable in \theoexprest (as we only need bounded restrictions of $\sin$ and $\cos$ to capture e.g. $e^{\ii \omega_i}$ up to time $t_0$). As an immediate corollary of \cref{thm: eventual cone contained in o-min invariant}, we have the following theorems.
\begin{theorem}
  \label{thm: invariant characterization}
  Let $\tup{A,\vect x_0}$ be a CDS. For every $t_0\ge 0$, the set $\cone_{t_0}\cup \set{e^{At}\vect x_0: 0\le t\le t_0}$ is an invariant that contains the whole orbit of $\tup{A,\vect x_0}$.  Moreover, this invariant is definable in $\theoexprest$ (and in particular is o-minimal).
\end{theorem}
\begin{theorem}
  \label{thm: invariant sound complete}
  Let $\tup{A,\vect x_0}$ be a CDS and let $Y\subseteq \RR^d$ be an error set. There exists an o-minimal invariant $\inv$ that contains the orbit and is disjoint from $Y$ if and only if there exists $t_0$ such that $\cone_{t_0}\cup \set{e^{At}\vect x_0: 0\le t\le t_0}$ is such an invariant.
\end{theorem}

\cref{thm: invariant sound complete} now allows us to provide an algorithm for deciding the existence of an invariant, subject to Schanuel's conjecture:
\begin{theorem}
  \label{thm: invariant decidable schanuel}
  Assuming Schanuel's conjecture, given a CDS $\tup{A,\vect x_0}$ and an $\theoexprest$ definable error set $Y$, it is
  decidable whether there exists an o-minimal invariant for $\tup{A,\vect x_0}$ that avoids $Y$. Moreover, if such an
  invariant exists, we can compute a representation of it.
\end{theorem}
\begin{proof}
  By \cref{thm: invariant sound complete}, there exists an o-minimal invariant $\inv$ that avoids $Y$ if and only if there exists some $t_0\in \RR$ such that $\cone_{t_0}\cup \set{e^{At}\vect x_0: 0\le t\le t_0}$ is such an invariant. Thus, the problem reduces to deciding the truth value of the following \theoexprest sentence:
\begin{align*}
  \exists t_0\st (\cone_{t_0}\cup \set{e^{At}\vect x_0: 0\le t\le t_0})\cap Y=\emptyset
\end{align*}
The theory of \theoexprest is decidable subject to Schanuel's conjecture, and therefore we can decide the existence of an invariant. Moreover, if an invariant exists, we can compute a representation of it by iterating over increasing values of $t_0$, until we find a value for which the sentence $\big(\cone_{t_0}\cup \set{e^{At}\vect x_0: 0\le t\le t_0}\big) \cap Y=\emptyset$ is true.
\end{proof}


%% file: fatcones.tex
\section{Semi-algebraic Error Sets and Fat Trajectory Cones}
\label{sec: fat}
In this section, we restrict attention to semi-algebraic invariants and semi-algebraic error sets, in order to regain unconditional decidability.

Substitute $s=e^t$ in the definition of the cone to get:
\begin{align*}
  \cone_{t_0} = \left\{P\ E(\log s)\ \diag(\vect \tau)\ Q(\log s)\ P^{-1}\vect x_0\st \vect \tau\in\torus_\omega,\ s\ge e^{t_0} \right\}.
\end{align*}
Written this way, observe that $E(\log s)=\diag(s^{\rho_1},\ldots,s^{\rho_k})$, which is almost semi-algebraic, apart from the fact that the exponents need not be rational.

\subsection{Unconditional Decidability}

We give the final, yet crucial property of the cones. When the error set is semi-algebraic, it is possible to decide, unconditionally, whether there exists some cone that avoids the error set. Moreover the proof is constructive, it will produce the cone for which this property holds. 
\begin{theorem}
  \label{thm: semi algebraic target decidable}
  For a semi-algebraic error set $Y$, it is (unconditionally) decidable whether there exists $t_0\ge 0$ such that $\cone_{t_0}\cap Y=\emptyset$. Moreover, such a $t_0$ can be computed.
\end{theorem}
\begin{proof}

  Define the set
  \begin{align*}
    U\defequals\left\{\cV\in\RR^{d\times d}\st \forall\vect\tau\in\torus_\omega,\ \ P\ \cV\ \diag(\vect\tau)\ P^{-1}\vect x_0\in\RR^d\setminus Y\right\}.
  \end{align*}
  The set $U$ can be seen to be semi-algebraic and thus is expressed by a quantifier-free formula that is a finite disjunction of formulas of the form $\bigwedge_{l=1}^m R_l(\cV)\sim_l 0$, where each $R_l$ is a polynomial with integer coefficients, over $d\times d$ variables of the entries of the matrix $\cV$, and$\sim_l\in \{>,=\}$. Define the matrix
  \begin{align*}
    \Lambda(s)\defequals \diag(s^{\rho_1},\ldots, s^{\rho_k}) Q(\log s)\in \RR^{d\times d},
  \end{align*}
  and notice that $\cone_{t_0}\cap Y=\emptyset$ if and only if $\Lambda(s)\in U$ for every $s\ge e^{t_0}$. 
Thus, it is enough to decide whether there exists $s_0\ge 1$ such that for every $s\ge s_0$, at least one of the disjuncts $\bigwedge_{l=1}^m R_l(\Lambda(s))\sim_l 0$ is satisfied.

Since $R_l(\Lambda(s))$ are polynomials in entries of the form $s^{\rho_i}$ and $\log(s)$, there is an effective bound $s_0$ such that for all $s\ge s_0$, none of the values $R_l(\Lambda(s))$ change sign for any $1\le l\le m$. Hence we only need to decide whether there exists some $s_0'\ge s_0$ such that for all $s\ge s_0'$ we have $R_l(\Lambda(s))\sim_l 0$ for every $1\le l\le m$.

Fix some $l$. The polynomial $R_l(v_1,\ldots,v_D)$ has the form $ \sum_i a_i v_1^{n_{i,1}}\cdots v_{D}^{n_{i,D}} $. After identifying the matrix $\Lambda(s)$ with a vector in $\RR^{D}$ for $D=d^2$, we see that $R_l(\Lambda(s))$ is a sum of terms of the form
\begin{align*}
  a_is^{n'_{i,1}\rho_1+\ldots n'_{i,k}\rho_k}\cdot Q_{i,1}(\log s)\cdots Q_{i,D}(\log s)
\end{align*}
where the $n'_{i,j}$ are aggregations of the $n_{i,j}$ for identical entries of $\diag(s^\rho_1,\ldots,s^\rho_k)$, and $Q_{i,j}(\log s)$ are polynomials obtained from the entries of $Q(\log s)$ under $R_l$. We can join the polynomials $Q_1,\ldots,Q_D$ into a single polynomial $f_i$, which would also absorb $a_i$. Thus, we rewrite $R_l$ in the form $ \sum_{i} s^{n'_{i,1}\rho_1+\ldots n'_{i,k}\rho_k} f_i(\log s) $ where each $f_i$ is a polynomial with rational coefficients (as the coefficients in $Q(\log s)$ are rational).

In order to reason about the sign of this expression as $s\to \infty$, we need to find the leading term of $R_l(\Lambda(s))$. This, however, is easy: the exponents $n'_{i,1}\rho_1+\ldots +n'_{i,k}\rho_k$ are algebraic numbers, and are therefore susceptible to effective comparison. Thus, we can order the terms by magnitude. Then, we can determine the asymptotic sign of each coefficient $f_i(\log s)$ by looking at the leading term in $f_i$.

We can thus determine the asymptotic behaviour of each $R_l(\Lambda(s))$, to conclude whether $\bigwedge_{l=1}^m R_l(\Lambda(s))\sim_l 0$ eventually holds. Moreover, for rational $s$, every quantity above can be computed to arbitrary precision, therefore it is possible to compute a threshold $s_0'$, after which, for all $s\ge s_0'$, $\bigwedge_{l=1}^m R_l(\Lambda(s))\sim_l 0$ holds. This completes the proof. 
\end{proof}

\begin{theorem}
  \label{thm:unconditional decid}
  For a semi-algebraic set $Y$, it is decidable whether there exists a o-minimal invariant, disjoint from $Y$, that contains the orbit $\orb(u)$ after some time $u\ge 0$. Moreover in the positive instances an invariant that is \theoexp-definable can be constructed.
\end{theorem}
\begin{proof}
  If there is an invariant $\inv$ that contains $\orb(u)$, for some $u\ge 0$, then \cref{thm: eventual cone contained in o-min invariant} implies that there exists some $t_0\ge u$ such that $\cone_{t_0}$ is contained in $\inv$. Consequently, the question that we want to decide is equivalent to the question of whether there exists a $t_0$, such that $\cone_{t_0}\cap Y=\emptyset$. The latter is decidable thanks to \cref{thm: semi algebraic target decidable}. The effective construction follows from the fact that such a $t_0$ is computable and that the cone is \theoexp-definable. 
\end{proof}

\subsection{Effectively Constructing the Semi-algebraic Invariant}

We now turn to show that in fact, for semi-algebraic error sets $Y$, we can approximate $\cone_{t_0}$ with a semi-algebraic set such that if $\cone_{t_0}$ avoids $Y$, so does the approximation. Intuitively, this is done by relaxing the ``non semi-algebraic'' parts of $\cone_{t_0}$ in order to obtain a \emph{fat cone}.
This relaxation has two parts: one is to ``rationalize'' the (possibly irrational) exponents $\rho_1,\ldots,\rho_k$, and the other is to approximate the polylogs in $Q(\log s)$ by polynomials.

\subparagraph*{Relaxing the exponents.}  We start by approximating the exponents $\rho_1,\ldots,\rho_k$ with rational numbers. We remark that naively taking rational approximations is not sound, as the approximation must also adhere to the additive relationships of the exponents.

Let $\vect \ell =(\ell_1,\ldots,\ell_k)$ and $\vect u = (u_1,\ldots,u_k)$ be tuples of rational numbers such that $\ell_i \leq \rho_i \leq u_i$ for $i=1,\ldots,k$.  Define $\bbS\subseteq \RR^k$ as:
\begin{align*}
  \bbS\defequals\left\{(q_1,\ldots,q_k)\in \RR^k \st \forall n_1,\ldots,n_k\in\ZZ,\ \ \left(
  \sum_{i=1}^k n_i \rho_i = 0 \Rightarrow \sum_{i=1}^k n_i q_i=0 \right )\right\}
\end{align*}
Thus, $\bbS$ captures the integer additive relationships among the $\rho_i$. Define
\begin{align*}
  \mathrm{Box}(\vect \ell,\vect u) \defequals \set{ \diag(\vect q) \st \vect \ell \leq \vect q \leq \vect u, \vect q \in \bbS}.
\end{align*}

\subparagraph*{Approximating polylogs.}  Let $\epsilon,\delta>0$. We simply replace $\log s$ by $r$ such that $\delta\le r\le s^\epsilon$. Note that it is not necessarily the case that $\delta\le \log s\le s^\epsilon$, so this replacement is a-priori not sound. However, for large enough $s$ the inequalities do hold, which will suffice for our purposes.

We can now define the fat cone. Let $\epsilon,\delta>0$ and $\lowervec =(\ell_1,\ldots,\ell_k)$ and $\uppervec =
(u_1,\ldots,u_k)$ as above, the fat orbit cone $\cF_{s_0,\epsilon,\delta,\vect \ell,\vect u}$ is the set: 
\begin{align*}
  \bigg\{P\ \diag(s^{q_1},\ldots,s^{q_k})\diag(\vect\tau)\ Q(r)P^{-1}\vect x_0
  \st \vect\tau\in\torus_\omega,\ s\geq s_0,\ \delta\le r\le s^\epsilon,\ \vect{q}
  \in \mathrm{Box}(\vect \ell,\vect u)\bigg\}.
\end{align*}
That is, the fat cone is obtained from $\cone_{t_0}$ with the following changes:
\begin{itemize}
	\item $R(\log s)=\diag(s^{\rho_1},\ldots,s^{\rho_k})$ is replaced with $\diag(s^{q_1},\ldots,s^{q_k})$, where the $q_i$ are rational approximations of the $\rho_i$, and maintain the additive relationships.
	\item $Q(\log s)$ is replaced with $Q(r)$ where $\delta\le r\le s^\epsilon$.
        \item The variable $s$ starts from $s_0$ (as opposed to $e^{t_0}$). 
\end{itemize}

We first show that the fat cone is semi-algebraic (the proof is in \cref{ap:proofs of sec4}), then proceed to prove that if there is a cone that avoids the error set, then there is a fat one that avoids it as well. 
\begin{restatable}{lemma}{fatconesemi}
	\label{lem: fat cone is semi-algebraic and we can compute a representation}
	$\cF_{s_0,\epsilon,\delta,\lowervec,\uppervec}$ is definable in \theoreals, and we can compute a representation of it.
\end{restatable}
\begin{lemma}
	\label{lem: fat cone is invariant for sa target}
	Let $Y\subseteq \RR^d$ be a a semi-algebraic error set such that $\cone_{t_0}\cap Y=\emptyset$ for some $t_0\in \RR$, then there exists $\delta,\epsilon,s_0,\boldsymbol{\ell},\boldsymbol{u}$ as above such that
	\begin{enumerate}
		\item $\cF_{s_0,\epsilon,\delta,\boldsymbol \ell,\boldsymbol u}\cap Y=\emptyset$, and
		\item for every $t\ge 0$ it holds that $e^{At}\cdot \cF_{s_0,\epsilon,\delta,\boldsymbol \ell,\boldsymbol u}\subseteq \cF_{s_0,\epsilon,\delta,\boldsymbol \ell,\boldsymbol u}$.
	\end{enumerate}
\end{lemma}

The result is constructive, so when $t_0$ is given, the constants $s_0, \epsilon, \delta, \vect \ell, \vect u$ can be computed. It follows that a corollary of this lemma, and \cref{lem: fat cone is semi-algebraic and we can compute a representation}, is a stronger statement than that of \cref{thm:unconditional decid}, namely one where \theoexp is replaced by \theoreals. We state it here before moving on with the proof of \cref{lem: fat cone is invariant for sa target}.

\begin{theorem}
  \label{thm:unconditional decid 2}
  For a semi-algebraic set $Y$, it is decidable whether there exists a o-minimal invariant, disjoint from $Y$, that contains the orbit $\orb(u)$ after some time $u\ge 0$. Moreover in the positive instances an invariant that is \theoreals-definable can be constructed.
\end{theorem}

The proof of~\cref{lem: fat cone is invariant for sa target} is given by the two corresponding steps. The second step, proving the invariance of the fat cone, is \cref{lem: fat cone eventual invariant for all epsilon} in \cref{ap:proofs of sec4}. We turn our attention to the first step. 
\begin{lemma}
	\label{lem: fat cone eventually disjoint}
	Let $Y\subseteq \RR^d$ be a semi-algebraic error set, and let $t_0\in \RR$ be such that $\cone_{t_0}\cap Y=\emptyset$, then there exists $\delta,\epsilon,s_0,\vect \ell,\vect u$ as above such that $\cF_{s_0,\epsilon,\delta,\vect \ell,\vect u}\cap Y=\emptyset$.
\end{lemma}
\begin{proof}
  We use the same analysis and definitions of $U$, $R_l$, $\sim_l$, $\Lambda(s)$ as in the proof of \cref{thm: semi algebraic target decidable} and focus on a single polynomial $R_l$. Recall that we had
  \begin{align}
    \label{eq: Rl semialgebraic inv}
    R_l(\Lambda(s))=\sum_{i} s^{n_{i,1}\rho_1+\ldots n_{i,k}\rho_k} f_i(\log s)
  \end{align}
  where each $f_i$ is a polynomial with rational coefficients.
	
  Denote $\vect \rho=(\rho_1,\ldots,\rho_k)$. We show, first, how to replace the exponents vector $\vect \rho$ by any exponents vector in $ \boxvec(\lowervec,\uppervec)$ for appropriate $\lowervec,\uppervec$, and second, how to replace $\log s$ by $r$ where $\delta\le r\le s^{\epsilon}$ for some appropriate $\delta$ and $\epsilon$, while maintaining the inequality or equality prescribed by $\sim_l$.
	
  Denote by $N$ the set of vectors $\vect{n_i}=(n_{i,1},\ldots,n_{i,{k}})$ of exponents in~\eqref{eq: Rl semialgebraic inv}. Let $\mu>0$, such that for every $\vect{n},\vect{n'}\in N$, if $\vect{\rho}\cdot (\vect{n}-\vect{n'})\neq 0$ then $|\vect{\rho}\cdot (\vect{n}-\vect{n'})|>\mu$. That is, $\mu$ is a lower bound on the minimal difference between distinct exponents in~\eqref{eq: Rl semialgebraic inv}. Observe that we can compute a description of $\mu$, as the exponents are algebraic numbers.

  Let $M=\max_{\vect{n,n'}\in N}\norm{\vect{n}-\vect{n'}}$ (where $\norm{\cdot}$ is the Euclidean norm in $\RR^k$).
  \begin{claim}
    \label{prop: approximate exponent}
    Let $\vect{c}\in \RR^k$ be such that $\norm{\vect{\rho}-\vect{c}}\le \frac{\mu}{2M}$, then, for all $\vect{n},\vect{n'}\in N$, if $\vect{\rho}\cdot (\vect{n}-\vect{n'})> 0$ then $\vect{c}\cdot (\vect{n}-\vect{n'})>\frac{\mu}{2}$.
  \end{claim}
  \begin{proof}[Proof of \cref{prop: approximate exponent}]
    Suppose that $\vect{\rho}\cdot (\vect{n}-\vect{n'})> 0$, then by the above we have $\vect{\rho}\cdot (\vect{n}-\vect{n'})> \mu$, and hence
\[
      \vect{c}\cdot (\vect{n}-\vect{n'})= \vect{\rho}\cdot (\vect{n}-\vect{n'})+(\vect{c}-\vect{\rho})\cdot (\vect{n}-\vect{n'})
     \ge  \mu-\norm{\vect{c}-\vect{\rho}}\cdot\norm{\vect{n}-\vect{n'}}
      \ge  \mu-\frac{\mu}{2M}M=\frac{\mu}{2}.
\]
  \end{proof}
  \newcommand{\vectau}{\boldsymbol\tau} We can now choose $\lowervec$ and $\uppervec$ such that $u_i-\ell_i\le \frac{\mu}{2M\sqrt{k}}$ and for all $\vect c\in \boxvec(\lowervec,\uppervec)$ we have
  \begin{align*}
	\norm{\vect{\rho}-\vect{c}}\le \sqrt{\sum_{i=1}^k (u_i-\ell_i)^2}\le \sqrt{\frac{\mu^2}{(2M)^2}}=\frac{\mu}{2M}.
  \end{align*}
  It follows from \cref{prop: approximate exponent} and from the definition of $\boxvec(\lowervec,\uppervec)$ that, intuitively, every $\vect{c}\in \boxvec(\lowervec,\uppervec)$ maintains the order of magnitude of the monomials $s^{n_{i,1} \cdot \rho_1+\ldots + n_{i,k} \cdot \rho_k}$ in $R_l(\Lambda(s))$.
	
	More precisely, let $\Lambda'(s)=\diag(s^{c_1},\ldots,s^{c_k})Q(\log s)$ for some $\vect{c}\in \boxvec(\lowervec,\uppervec)$, then the exponent of the ratio of every two monomials in $R_l(\Lambda'(s))$ has the same (constant) sign as the corresponding exponent in $R_l(\Lambda(s))$. Moreover, the exponents of distinct monomials in $R_l(\Lambda(s))$ differ by at least $\frac{\mu}{2}$ in $R_l(\Lambda'(s))$.
	
	We now turn our attention to the $\log s$ factor. First, let $s_0$ be large enough that $f_i(\log s)$ has constant sign for every $s\ge s_0$.  We can now let $\delta$ be large enough such that for every $r\ge \delta$, the sign of $f_i(\log s)$ coincides with the sign of $f_i(r)$ for every $s\ge s_0$.  It remains to give an upper bound on $r$ of the form $s^\epsilon$ such that plugging $f_i(r)$ instead of $f_i(\log s)$ does not change the ordering of the terms (by their magnitude) in $R_l(\Lambda'(s))$.
	
	Let $B$ be the maximum degree of all polynomials $f_i$ in~\eqref{eq: Rl semialgebraic inv}, and define $\epsilon=\frac{\mu}{3B}$ (in fact, any $\epsilon<\frac{\mu}{2B}$ would suffice), then we have that, for $s\ge s_0$, $f_i(r)$ has the same sign as $f_i(\log s)$ for every $\delta\le r\le s^{\epsilon}$ (by our choice of $\delta$), and guarantees that plugging $s^{\epsilon}$ instead of $s$ does not change the ordering of the terms (by their magnitude) in $R_l$.  Since the exponents of the monomials in $R_l(\Lambda'(s))$ differ by at least $\frac{\mu}{2}$, it follows that their order is maintained when replacing $\log s$ by $\delta\le r\le s^\epsilon$.

	Let $\Lambda''(s)=\diag(s^{c_1},\ldots,s^{c_k})Q(r)$ for some $\vect{c}\in \boxvec(\lowervec,\uppervec)$ and $\delta\le r\le s^\epsilon$, then by our choice of $\epsilon$, the dominant term in $R_l(\Lambda''(s))$ is the same as that in $R_l(\Lambda(s))$. Therefore, for large enough $s$, the signs of $R_l(\Lambda''(s))$ and $R_l(\Lambda(s))$ are the same.
	
	Note that since $\cone_{t_0}\cap Y=\emptyset$, then w.l.o.g. $R_l(\Lambda(s))\sim_l 0$ for every $l$. Thus, by repeating the above argument for each $R_l$, we can compute $s_0\in \RR$, $\epsilon>0$, $\delta\in \RR$, and $\lowervec,\uppervec\in \QQ^k$ such that 
	$\cF_{s_0,\epsilon,\delta,\lowervec,\uppervec}\cap Y=\emptyset$, and we are done.
\end{proof}

%% file: hardness.tex
\section{A Reduction from Zeros of an Exponential Polynomial}
\label{sec:hardness}

In \cref{thm:unconditional decid 2}, we showed unconditional decidability for the question of whether there exists an invariant containing the orbit $\orb(u)$, for some $u\ge 0$. Even though we construct such an invariant, it cannot be used as a certificate proving that the orbit never enters the error set; however it is a certificate that the orbit of the system does not enter $Y$ {\em after} time $u$.

In this section we give indications that deciding whether there exists
an invariant that takes into account the orbit $\le u$ is
difficult. More precisely, we will reduce a problem about zeros of a
certain exponential polynomial to the question of whether there exists
a semi-algebraic invariant disjoint from $Y$ containing $\orb(0)$.

\begin{remark}
\label{rmk: discrete tail}	
In the setting of discrete linear dynamical systems, the existence of a semi-algebraic invariant from time $t_0$ immediately implies the existence of one from time $0$. This is because the system goes through finitely many points from $0$ to $t_0$, which can be added one by one to the semi-algebraic set. In this respect CDSs are more complicated to analyse. 
\end{remark}

The problem that we reduce from, can be stated as follows. We are
given as input real algebraic numbers
$a_1,\ldots,a_n,\rho_1,\ldots,\rho_n$, and $t_0\in\QQ$, and asked to
decide whether the exponential function:
\begin{align*}
  f(t)\defequals a_1e^{\rho_1t}+\cdots +a_ne^{\rho_nt},
\end{align*}
has any zeros in the interval $[0,t_0]$.  This is a special case
of the so-called Continuous Skolem Problem~\cite{Bell10,chonev2016skolem}.

While there has been progress on characterising the asymptotic
distribution of complex zeros of such functions, less is known about
the real zeros, and we lack any effective characterisation, see
\cite{Bell10,chonev2016skolem} and the references therein. The
difficulty of knowing whether $f$ has a zero in the specified region
is because (a) all the zeros have to be transcendental (a consequence
of Hermite-Lindemann Theorem) and (b) there can be tangential zeros,
that is $f$ has a zero but it never changes its sign. See the
discussion in \cite[Section 6]{Bell10}. Finding the zeros of such a
polynomial is a special case of the {\em bounded} continuous
Skolem problem. We note that when $\rho_i$ are all rational the
problem is equivalent to a sentence of \theoreals (and hence
decidable) by replacing $t=\log s$.

The rest of this section is devoted to the proof of the following theorem.

\begin{theorem}
  \label{thm:main hardness}
  For every exponential polynomial $f$ we can construct a CDS $\tup{A,\vect x_0}$ and semi-algebraic set $Y$ such that the following two statements are equivalent:
  \begin{itemize}
  \item there exists a semi-algebraic invariant disjoint from $Y$ that contains $\orb(0)$,
  \item $f$ does not have a zero in $[0,t_0]$. 
  \end{itemize}
\end{theorem}

Fix the function $f$, {\em i.e.} real algebraic numbers $a_1,\ldots, a_n,\rho_1,\ldots,\rho_n$ and $t_0\in\QQ$. Without loss of generality we can assume that $\rho_1,\ldots,\rho_n$ are all nonnegative, since $e^{\rho t}f(t)=0$ if and only if $f(t)=0$ where $\rho$ is larger than all $\rho_1,\ldots,\rho_n$. 

Since every $\rho_i$ is algebraic, there is a minimal polynomial $p_i$, that has $\rho_i$ as a simple root. Let $A$ be the $d\times d$ companion matrix of the polynomial $p_1(x)\cdots p_n(x)x^2$. The numbers $\rho_i$ are eigenvalues $A$ of multiplicity one, and the latter also has zero as an eigenvalue of multiplicity two. In addition to those, the matrix $A$ generally has other (complex) eigenvalues as well. We put $A$ in Jordan normal form, $P^{-1}AP=J$ where $J$ is made of two block diagonals: $\tilde A$ and $B$, where
\begin{align*}
  \tilde A\defequals \begin{pmatrix}
\diag(\rho_1,\ldots,\rho_n) & &\\
& 0 & 1\\
& 0 & 0
\end{pmatrix},
\end{align*}
and $B$ is some $(d-n-2)\times (d-n-2)$ matrix. Define:
\begin{align*}
  \vect {\tilde x}_0\defequals (\underbrace{1,\ldots, 1}_{n+2},0,\ldots,0),
\end{align*}
the vector that has $n+2$ ones and the rest, $d-(n+2)$ zeros, whose purpose is to ignore the contribution of the eigenvalues in matrix $B$ in the system. To simplify notation, since $\vect {\tilde x}_0$ is ignoring the contribution of the matrix $B$, the dynamics of the system $\tup{J,\vect {\tilde x}_0}$ can be assume to be the same as:
\begin{align*}
  e^{\tilde At}(1,\ldots,1)=(e^{\rho_1t},\ldots,e^{\rho_nt},t). 
\end{align*}

Focus on a single eigenvalue, {\em i.e.} on the graph $\set{(e^{\rho t},t)\st t\ge 0}$, as the analysis will easily generalise to the CDS in question. This is itself a CDS, so terminology such as orbits {\em etc.} make sense. The challenge is to find a family of {\em tubes} around  this exponential curve such that (a) all the tubes together with $\set{(y,t)\st t\ge t_0}$ are invariants and (b) the tubes are arbitrarily close approximations of the curve. 

We achieve this by the following families of polynomials:
\begin{itemize}
\item under-approximations are given by the family indexed by $n\in\NN$:
  \begin{align*}
    P_n(t)\defequals\sum_{k=0}^n\frac{(\rho t)^k}{k!}. 
  \end{align*}
\item over-approximations are given by a family indexed by $n\in\NN$ and $\mu>1$:
  \begin{align*}
    Q_{n,\mu}(t)\defequals P_n(\mu t).
  \end{align*}
\end{itemize}
Define:
\begin{align*}
  \inv_{n,\mu}\defequals \left\{(y,t)\st P_n(t)\le y\le Q_{n,\mu}(t)\text{ and }0\le t\le t_0\right\}.
\end{align*}

It is clear from Taylor's theorem and the assumption that $\rho>0$, that by taking $n\to\infty$, and $\mu\to 1^+$ the sets $\inv_{n,\mu}$ are arbitrary precise approximations of the graph $\set{(e^{\rho t},t)\st t\ge 0}$, what remains to show is that they are invariant.

\begin{restatable}{lemma}{lemmahardness}
  \label{lem:hardness invariant}
  For every $\mu>1$ there exists $n_0\in\NN$ such that for all $n\ge n_0$ the set
  \begin{align*}
    \inv_{n,\mu}\cup\set{(y,t)\st t>t_0}
  \end{align*}
  is an invariant containing the whole orbit, {\em i.e.} $\set{(e^{\rho t},t)\st t\ge 0}$.
\end{restatable}

The proof is in \cref{ap:proofs of sec5}.

We can construct such invariants for every curve $e^{\rho_i t}$, and thus build $\tilde\inv_{n,\mu}$ for
\begin{align*}
  \left\{ (e^{\rho_1 t},\ldots,e^{\rho_n t},t)\st t\ge 0\right\}. 
\end{align*}
To prove \cref{thm:main hardness} we define $\tilde Y$ by the formula
\begin{align*}
  \Phi(x_1,\ldots,x_n,x_{n+1})\defequals a_1x_1 +\cdots + a_nx_n = 0\text{ and }0\le x_{n+1}\le t_0.
\end{align*}

Since the analysis was done on the CDS $\tup{J,\vect {\tilde x}_0}$, whose entries are not rational in general, before proceeding with the proof of \cref{thm:main hardness}, we need the following lemma to say that changing basis does not have an effect in the decision problem at hand:
\begin{restatable}{lemma}{lemchangebasis}
  \label{lem:changebasis}
  For every $\tilde Y$ semi-algebraic, there exists another semi-algebraic set $Y$ and $\vect x_0$ with rational entries
  such that the following two statements are equivalent:
  \begin{itemize}
  \item $\tup {J,\vect {\tilde x}_0}$ has a semi-algebraic invariant disjoint from $\tilde Y$, containing the whole orbit,
  \item $\tup {PJP^{-1},\vect x_0}$ has a semi-algebraic invariant disjoint from $Y$, containing the whole orbit.
  \end{itemize}
\end{restatable}
The proof is postponed to \cref{ap:proofs of sec5}. Thanks to this lemma, we can prove \cref{thm:main hardness} for the CDS $\tup{J, \vect {\tilde x}_0}$ and the set $\tilde Y$ instead. This is done as follows. The direct implication is trivial. For the converse, observe that $f(t)$ does not have a zero in $[0,t_0]$ if and only if the $\orb(0)$ and $\tilde Y$ are disjoint. Since both $\orb(0)$ and $\tilde Y$ are closed sets, we can find a tube that contains $\orb(0)$ and is disjoint from $\tilde Y$, {\em i.e.} there exists some $\mu>1$ and $n\in\NN$ such that
\begin{align*}
  \tilde\inv_{n,\mu} \cup \set{(y,t)\st t>t_0}, 
\end{align*}
is an invariant that is disjoint from $\tilde Y$ but contains $\orb(0)$. 


%% file: proofProp.tex
\section{Proof of \cref{prop: real entries}}
\label{apx: cone is real}
\propreals*
Write $P = \begin{pmatrix} P_1 & \cdots & P_k \end{pmatrix}$ with $P_i$ having dimension $d\times d_i$ for $i\in\{1,\ldots,k\}$.  The condition $A=PJP^{-1}$ is equivalent to $AP=PJ$, which in turn is equivalent to $AP_i = P_iJ_i$ for $i=\{1,\ldots,k\}$. Now if $AP_i = P_iJ_i$ then $A\overline{P_i}=\overline{P_i}\overline{J_i}$ and hence we may assume without loss of generality that for ${i_1},i_2\in\{1,\ldots,k\}$, if $\overline{J_{i_1}}=J_{i_2}$ then $\overline{P_{i_1}}=P_{i_2}$.  Equivalently we may assume that $\overline{P}=PM$ for $M$ a permulation matrix that interchanges column $(i_1,j)$ of $P$ with column $(i_2,j)$ such that $\overline{J_{i_1}}={J_{i_2}}$.  Then we have
\begin{eqnarray*}
  \overline{P\, 
  \diag(B_1,\ldots,B_k)P^{-1}} 
  &=& \overline{P}\, \diag(\overline{B_1},\ldots,\overline{B_k}) \overline{P}^{-1} \\
  &=& PM \diag(\overline{B_1},\ldots,\overline{B_k}) M^{-1}P^{-1} \\
  &=& P\diag(B_1,\ldots,B_k) P^{-1} \, .
\end{eqnarray*}
Hence $P\, \mathrm{diag}(B_1,\ldots,B_k)P^{-1}$ is real.
\qed


%% file: containscone.tex
\section{Proof of \cref{thm: eventual cone contained in o-min invariant}}
\label{apx: invariant contains cone}
\eventualcone*
Before proceeding with the proof, we give some useful definitions and properties of o-minimal theories. Consider an o-minimal theory $\theo$.

A function $f\colon B\to \RR^m$ with $B\subseteq \RR^n$ is \emph{definable} in $\theo$ if its graph $\Gamma(f)=\set{(\vect x,f(\vect x)) \st \vect x\in B}\subseteq \RR^{n+m}$ is an $\theo$-definable set.

O-minimal theories admit the following properties (see~\cite{dries_1998} for precise definitions and proofs).
\begin{enumerate}
	\item
	For an $\theo$-definable set $S\subseteq \RR^d$, its topological closure $\overline{S}$ is also $\theo$-definable.
	\item
	For an $\theo$-definable function $f \colon S\to \RR$, the number $\inf\set{f(\vect x) \st \vect x\in S}$ is $\theo$-definable (as a singleton set).
	\item
	O-minimal structures admit \emph{cell decomposition}: every $\theo$-definable set $S\subseteq \RR^d$ can be written as a finite union of connected components called \emph{cells}. Moreover, each cell is $\theo$-definable and homeomorphic to $(0,1)^m$ for some $m \in \{0, 1, \ldots, d\}$ (where for $m=0$ we have that $(0,1)^0$ is a single point, namely $\{\vect{0}\}\subseteq \RR^d$). The \emph{dimension} of $S$ is defined as the maximal such $m$ occurring in the cell decomposition of $S$. 
	\item
	For an $\theo$-definable function $f \colon S\to \RR^m$, the dimension of its graph $\Gamma(f)$ is the same as the dimension of $S$.
\end{enumerate}

We recall the definition of the orbit cone:
\begin{align*}
  \cone_{t_0} \defequals \left\{P\ E(t)\ \diag(\tau)\ Q(t)\ P^{-1}\vect x_0\st \tau\in\torus_\omega,t\geq t_0 \right\},
\end{align*}
and define the {\em orbit rays} for $\tau\in\torus_\omega$:
\begin{align*}
  \ray(\tau,t_0) \defequals \left\{P\ E(t)\ \diag(\tau)\ Q(t)\ P^{-1}\vect x_0\st t\geq t_0 \right\}.
\end{align*}
Fix $\inv$ to be an o-minimal invariant, with $\orb\subseteq \inv$ definable in \theo. To prove Theorem~\ref{thm: eventual cone contained in o-min invariant}, we begin by making following claims of increasing strength:
\begin{claim}
  \label{clm:ev-in-or-out}
  For every $\vect \tau \in \torus_\omega$ there exists $t_0 \ge 0$ such that $\ray(\vect \tau,t_0)\subseteq \inv$ or
  $\ray(\vect \tau,t_0)\cap \inv=\emptyset$.
\end{claim}
\begin{claim}
  \label{clm:ev-in}
  For every $\vect \tau \in \torus_\omega$ there exists $t_0 \ge 0$ such that $\ray(\vect \tau,t_0)\subseteq \inv$.
\end{claim}
\begin{claim}
  \label{clm:in-uniform}
  There exists $t_0 \ge 0$ such that for every $\vect \tau\in \torus_\omega$ we have $\ray(\vect \tau,t_0)\subseteq \inv$.
\end{claim}
\begin{proof}[Proof of \cref{clm:ev-in-or-out}]
  Fix $\vect \tau \in \torus_\omega$.  Then the set
  \begin{align*}
    \set{t\ge 0 \st P\ E(t)\ \diag(\vect \tau)\ Q(t)P^{-1}\vect x_0\in\inv}
  \end{align*}
  is $\theo$-definable and hence comprises a finite union of intervals.  If this set contains an unbounded interval then
  there exists $t_0$ such that $\ray(\vect \tau,t_0)\subseteq \inv$; otherwise there exists $t_0$ such that
  $\ray(\vect \tau,t_0)\cap \inv=\emptyset$.
\end{proof}

\begin{proof}[Proof of \cref{clm:ev-in}]
  We strengthen \cref{clm:ev-in-or-out}.  Assume by way of contradiction that there exist $\vect \tau\in \torus_\omega$ and $t_0\in \RR$ such that $\ray(\vect \tau,t_0)\cap \inv=\emptyset$. Without loss of generality assume that $t_0>1$, and consider $e^{-A} \cdot \ray(\vect \tau,t_0)$. Recall from analysis of $e^{At}$ the decomposition:
  \begin{align*}
    e^{-A}= P\ E(-1)\ R(-1)\ Q(-1)\ P^{-1},
  \end{align*}
  and let $\vect \tau'\in\torus_\omega$ be equal to $R(-1)\diag(\vect \tau)$. In other words, $\diag(\vect \tau)=\diag(\vect \tau')R(1)$ and
  hence $e^{A}\ray(\vect \tau',t_0-1)=\ray(\vect \tau,t_0)$ (this is implicitly shown in the proof of \cref{lem: cone is invariant}). Since $\inv$ is invariant we have $\ray(\vect \tau',t_0-1)\cap\inv=\emptyset$, and consequently $\ray(\vect \tau',t_0)$ itself is disjoint from $\inv$. 

  Repeating this argument, we get that for every $n\in \NN$, the point $\diag(\vect \sigma)=R(-n)\diag(\vect \tau)$ satisfies
  $\ray(\vect \sigma,t_0)\cap \inv=\emptyset$.
	
  Let $U=\set{R(-n)\diag(\vect \tau) \st n\in \NN}$. Then diagonals of $U$ are dense in $\torus_\omega$, since the group of multiplicative relations defined by the $\{e^{-\ii\omega_1},\ldots ,e^{-\ii\omega_k}\}$ is the same as the one defined by $\{e^{\ii\omega_1},\ldots ,e^{\ii\omega_k}\}$. Set $U'=\set{\vect \sigma\in \torus_\omega \st \ray(\vect \sigma,t_0)\cap \inv=\emptyset}$ which is $\theo$-definable, and further, we have $U\subseteq U'\subseteq \torus_\omega$. Moreover, $\overline{U}=\torus_\omega$, so $\overline{U'}= \torus_\omega$.
	
  We now prove that, in fact, $U' = \torus_\omega$.  Assuming (again by way of contradiction) that there exists $\vect \sigma\in \torus_\omega\setminus U'$, then by the definition of $U'$ we have $\ray(\vect \sigma,t_0)\cap \inv\neq \emptyset$. It follows that for every $n\in \NN$, the point $\diag(\vect \sigma')=R(n)\diag(\vect \sigma)$ also satisfies $\ray(\vect \sigma',t_0)\cap \inv\neq \emptyset$. Define $V=\set{R(n)q \st n\in \NN}$, then the diagonals of $V$ are dense in $\torus_\omega$. Further the set $V'=\set{\vect \sigma'\in \torus_\omega \st \ray(\vect \sigma',t_0)\cap \inv\neq \emptyset}$ satisfies $V\subseteq V'\subseteq \torus_\omega$ and $\overline{V'}=\torus_\omega$. Now the sets $U'$ and $V'$ are both definable in \theo, and the topological closure of each of them is $\torus_\omega$.
	
  We employ \cite[Lemma 10]{ACOW18}, which states that if $X, Y \subseteq \torus_\omega$ are $\theo$-definable sets such that $\overline{X} = \overline{Y} = \torus_\omega$, then $X \cap Y \ne \emptyset$.
	
  It follows that $V'\cap U'\neq \emptyset$, which is clearly a contradiction.  Therefore, there is no $\vect \sigma \in \torus_\omega \setminus U'$; that is, $U' = \torus_\omega$.
	
  From this, however, it follows that $\cone_{t_0}\cap \inv=\emptyset$, which is again a contradiction, since $\cone_{t_0} \cap \orb \neq \emptyset$ and $\orb\subseteq \inv$, so we are done.
\end{proof}

\begin{proof}[Proof of \cref{clm:in-uniform}]
  Consider the function $f:\torus_\omega\to \RR$ defined by $f(\vect \tau)=\inf\{t\in \RR \st \ray(\vect \tau,t)\subseteq \inv\}$. By Claim~\ref{clm:ev-in} this function is well-defined. Since $\ray(\vect \tau,t)$ is $\theo$-definable, then so is $f$. Moreover, its graph $\Gamma(f)$ has finitely many connected components, and the same dimension as $\torus_\omega$. Thus, there exists an open set $K\subseteq \torus_\omega$ (in the induced topology on $\torus_\omega$) such that $f$ is continuous on $K$. Furthermore, $K$ is homeomorphic to $(0,1)^{m}$ for some $0\le m\le k$, and thus we can find sets $K''\subseteq K'\subseteq K$ such that $K''$ is open, and $K'$ is closed.\footnote{In case $m=0$, the proof actually follows immediately from Claim~\ref{clm:ev-in}, since $\torus_\omega$ is finite.}  Since $f$ is continuous on $K$, it attains a maximum on $K'$. Consider the set $\set{R(n)\cdot K'' \st n\in \NN}$. By the density of the diagonals of $\set{R(n) \st n\in\NN}$ in $\torus_\omega$, this is an open cover of $\torus_\omega$, and hence there is a finite subcover $\set{R(n_1)K'',\ldots,R(n_a)K''}$. Since $K''\subseteq K'$, it follows that $\set{R(n_1)K',\ldots,R(n_a)K'}$ is a finite closed cover of $\torus_\omega$.
	
  We now show that, for all $\vect \tau \in \torus_\omega$, we have $f(R(1) \vect \tau) \le f(\vect \tau)+1$. Indeed, consider any $\vect \tau\in \torus_\omega$ and $t > 0$ such that $\ray(\vect \tau, t) \subseteq \inv$. Applying $e^{A}$, we get $e^{A} \cdot \ray(\vect \tau, t) \subseteq e^{A} \inv \subseteq \inv$. Similarly to the proof of \cref{lem: cone is invariant}, we have that $e^{A} \cdot \ray(\vect \tau,t)= \ray(R(1) \vect \tau, t+1)$, so we can conclude that $\ray(R(1) \vect \tau, t+1) \subseteq \inv$. This means that $\ray(\vect \tau, t) \subseteq \inv$ implies $\ray(R(1) \vect \tau, t+1) \subseteq \inv$; therefore, $f(R(1) \vect \tau) \le 1+f(\vect \tau)$.
	
	Now denote $s_0=\max_{\vect \tau\in K'} f(\vect \tau)$. Then for every $1\le i\le m$ we have $\max_{\vect \tau\in R(n_i)K'}f(\vect \tau)\le n_i+ s_0$; so $f(\vect \tau)$ is indeed bounded on $\torus_\omega$.
\end{proof}
	Finally, we conclude from Claim~\ref{clm:in-uniform} that there exists $t_0\ge 0$ such that $\cone_{t_0}\subseteq \inv$. This completes the proof of \cref{thm: eventual cone contained in o-min invariant}.


%% file: proofsSec4.tex
\section{Proofs of \cref{sec: fat}}
\label{ap:proofs of sec4}
\fatconesemi*
\begin{proof}
  The only part that is not immediately semi-algebraic is the $\diag(s^{q_1},\ldots,s^{q_k})$ factor, as the exponents are not fixed.
	
	Consider the group $L\defequals\set{(n_1,\ldots, n_k)\in \ZZ^k\st \sum_{i=1}^k n_i\rho_i=0}$. Similarly to the analysis in \cref{sec: orbit cones}, we can compute a finite basis $\{\vect z^1,\ldots,\vect z^m\}\subseteq \ZZ^k$ for $L$. Then, we can rewrite $\bbS$ as $\bbS= \{(q_1,\ldots,q_k)\st \bigwedge_{j=1}^m q_1z^j_1+\ldots q_kz^j_k=0\}$. Next, observe that 
\begin{align*}
  \left\{\diag(s^{q_1},\ldots,s^{q_k})\st (q_1,\ldots,q_k)\in \bbS \right\}=\left\{\diag(w_1,\ldots,w_k)\st \bigwedge_{j=1}^m w_1^{z^j_1}\cdots w_k^{z^j_k}=1 \right\}.
\end{align*}
Indeed, for every $\vect z^j$ and $(q_1,\ldots,q_k)\in \bbS$ we have $(s^{q_1})^{z^j_1}\cdots (s^{q_k})^{z^j_k}=s^{q_1 z^j_1+\ldots +q_k {z^j_k}}=s^0=1$, and conversely, if $w_1,\ldots,w_k$ satisfy the condition on the right hand set, then for every $(n_1,\ldots,n_k)\in L$ we have $w_1^{n_1}\cdots w_k^{n_k}=1$, denote $q_i=\log_s w_i$, then this can be rewritten as $s^{q_1 n_1}\cdots s^{q_k n_k}=1$, so $n_1 q_1+\ldots+n_k q_k=0$, and hence $(q_1,\ldots ,q_k)\in \bbS$.
	
	Furthermore, the requirement $(q_1,\ldots,q_k)$ can be restated in the above formulation as $\ell_i\le \log_s w_i\le u_i$, or equivalently, $s^{\ell_i}\le w_i\le s^{u_i}$ (where $\lowervec=(\ell_1,\ldots,\ell_k)$ and $\uppervec=(u_1,\ldots,u_k)$).
	
Thus, define
\begin{align*}
  \bbS^{\diag}(L,U)\defequals\left\{\diag(w_1,\ldots,w_k)\st \bigwedge_{j=1}^m w_1^{z^j_1}\cdots w_k^{z^j_k}=1\text{ and for all }i,\ L_i\le w_i\le U_i \right\},
\end{align*}
then we can rewrite the fat cone as $\cF_{s_0,\epsilon,\delta,\boldsymbol \ell,\boldsymbol u}$ as the set
\begin{align*}
  \left\{P\ W\ \diag(\tau)\ Q(r)\ P^{-1}\vect x_0\st \vect\tau\in\torus_\omega,\ s\geq s_0,\ \delta\le r\le s^\epsilon,\ W
	\in \bbS^{\diag}(s^{\lowervec},s^{\uppervec})\right\}
\end{align*}
which is clearly semi-algebraic, and is equivalent by the above.
\end{proof}

\begin{lemma}
	\label{lem: fat cone eventual invariant for all epsilon}
For every $\epsilon>0$, there exists $s_0$ such that for every $s_1\ge s_0$, $t\ge 0$ and $\delta,\vect \ell,\vect u$ we have that $e^{At}\cF_{s_1,\epsilon,\delta,\vect \ell,\vect u}\subseteq \cF_{s_1,\epsilon,\delta,\vect \ell,\vect u}$
\end{lemma}
\begin{proof}
  Consider a vector
  \begin{align*}
    \vect v\defequals P\ \diag(s^{q_1},\ldots s^{q_k})\ \diag(\vect\tau)\ Q(r)\ P^{-1}\vect x_0\in \cF_{s_1,\epsilon,\delta,\vect \ell,\vect u},
  \end{align*}
  where $s_1$ will be determined later, and let $t\ge 0$. Set $t=\log x$ and recall that
  \begin{align*}
    e^{At}=e^{A\log x}=P\ \diag(x^{\rho_1},\ldots,x^{\rho_k})\ \diag(e^{\ii\omega_1 \log x},\ldots,e^{\ii\omega_k \log x})\ Q(\log x)\ P^{-1},
  \end{align*}
  whence
  \begin{align*}
    e^{At}v&=e^{A\log x}\vect v\\
           &=P\ \diag(x^{\rho_1}s^{q_1},\ldots, x^{\rho_k}s^{q_k})\diag(e^{\ii\omega_1 \log x}\tau_1,\ldots,e^{\ii\omega_k \log x}\tau_k)\ Q(\log x)\ Q(r)\ P^{-1}\vect x_0.
  \end{align*}
  We will now show that $e^{At}\vect v\in \cF_{s_1,\epsilon,\delta,\boldsymbol \ell,\boldsymbol u}$, by drawing some condition on $s_1$.  First, we claim that $(e^{\ii\omega_1 \log x}\tau_1,\ldots,e^{\ii\omega_k \log x}\tau_k)\in \torus_\omega$. Indeed, for all $j$ we have $|e^{\ii\omega_j \log x}\tau_j|=1$, and for all $\vect z$ such that $z_1\omega_1+\ldots+z_k\omega_k=0$, we have
  \begin{align*}
   (e^{\ii\omega_1 \log x}\tau_1)^{z_1}\cdots (e^{\ii\omega_k \log x}\tau_k)^{z_k}=e^{\ii\log x (z_1\omega_1+\ldots+z_k\omega_k)}\cdot \tau_1^{z_1}\cdots \tau_k^{z_k}=1 
  \end{align*}
  since $\vect \tau\in \torus_\omega$.
	
  Next, it is also not hard to prove that $(x^{\rho_1}s^{q_1},\ldots, x^{\rho_k}s^{q_k})$ can be written as
  \begin{align*}
    ((xs)^{p_1}, \ldots, (xs)^{p_k})
  \end{align*}
  for $(p_1,\ldots, p_k)\in \mathrm{Box}(\boldsymbol \ell,\boldsymbol u)$. Indeed, take $p_i=\frac{\rho_i\log x+q_i \log s}{\log x+ \log s}$, then for all $i$, $(xs)^{p_i}=\exp((\log x+\log s)p_i)=\exp(\rho_i\log x+q_i\log s)=x^{\rho_i} s^{q_i}$.
	
	It remains to show that $Q(\log x)\cdot Q(r)$ can be written as $Q(y)$ for $\delta\le y\le (xs)^\epsilon$.
	Recall that $Q(\log x)\cdot Q(r)=Q(\log x+r)$, and that $\delta\le r\le s^{\epsilon}$ and $x\ge 1$. It immediately follows that $\delta<\log x+r$. 

	Now, observe that $\log x+r\le \log x+s^\epsilon$. We prove that if $s_1$ is large enough, then $\log x+s^\epsilon\le (xs)^\epsilon$. Let $x_0\ge 1$ be such that for every $y\ge x_0$ we have $y^\epsilon\ge \max\{\log y,2\}$. Clearly such $x_0$ exists. We now split the proof into two cases.
	\begin{itemize}
		\item If $x> x_0$, take $s_1$ to be large enough such that $s^\epsilon\ge 2$ for every $s\ge s_1$. Then by the condition on $x_0$ we have that
		\[
		\log x+s^\epsilon\le x^\epsilon+s^\epsilon\le (xs)^\epsilon
		\]
		where the last inequality follows since both summands are at least $2$ (indeed, if $A,B\ge 2$ and w.l.o.g. $A\le B$, then $A+B\le 2 B\le AB$).
		
		\item If $x\le x_0$, recall that $x\ge 1$, and thus $\log x\le x-1$. So it suffices to find $s_1$ such that for all $s\ge s_1$ we have $x-1+s^\epsilon\le x^\epsilon s^\epsilon$. The latter is equivalent to $x-1\le (x^\epsilon-1)s^\epsilon$.
		
		Now, if $x=1$, the inequality holds for any $s$, and we are done. Otherwise, let $x>1$, then observe that the function $\frac{x-1}{x^\epsilon-1}$ is increasing, and $\lim_{x\to 1^+} \frac{x-1}{x^\epsilon-1}=\frac{1}{\epsilon}$ (e.g., by L'H\^opital's rule). In particular, the function $\frac{x-1}{x^\epsilon-1}$ is bounded from above on the interval $(0,x_0]$. Set $s_1$ be large enough such that for every $s\ge s_1$ and for every $x\in (0,x_0]$ we have $\frac{x-1}{x^\epsilon-1}\le s^\epsilon$, and we are done.
	\end{itemize}
By taking the maximal $s_1$ from the conditions above, we conclude the lemma.
\end{proof}


%% file: proofsSec5.tex
\section{Proofs of \cref{sec:hardness}}
\label{ap:proofs of sec5}
\lemmahardness*
To prove this lemma, we gather some properties of the under and over approximations. We recall their definitions here.

\begin{align*}
   &P_n(t)\defequals\sum_{k=0}^n\frac{(\rho t)^k}{k!},\\
   &Q_{n,\mu}(t)\defequals P_n(\mu t).
\end{align*}

\begin{proposition}
  \label{prop: underapproximations}
  The under-approximations have the following properties:
  \begin{itemize}
		\item \textbf{Property 1:} for all $n\in \NN$ and $0\le t\le t_0$, we have $P_n(t)\le e^{\rho t}$,
		\item \textbf{Property 2:} for all $n\in \NN$ and $0<t_1\le t\le t_0$, we have ${P_n}'(t)\le (P_n(t_1)e^{\rho (t-t_1)})'$, 
		\item \textbf{Property 3:} $\max_{0\le t\le t_0}\lVert P_n(t)-e^{\rho t}\rVert\to 0$ as $n\to\infty$.
   \end{itemize}
\end{proposition}
\begin{proof}
	Property 3 is satisfied by Taylor's theorem. Property 1 holds since $\rho>0$ by our assumption, in which case every Taylor polynomial of $e^{\rho t}$ is an under-approximation. 
	
	We turn to establish Property 2, which is equivalent to $P'_n(t)\le \rho P_n(t_1)e^{\rho(t-t1)}$. Note that it clearly holds for $n=0$.  Observe that $P'_n(t)=\rho P_{n-1}(t)$, thus we want to prove that $\rho P'_n(t)\le \rho P_n(t_1)e^{\rho(t-t_1)}$. Since $\rho>0$, we can cancel it from the inequality. Now consider the function $g_n(t)=P_n(t_1)e^{\rho(t-t_1)}-P_{n-1}(t)$, we prove that $g_n(t)\ge 0$ for all $t_1\le t\le t_0$.  First, we have that $g_n(t_1)=P_n(t_1)-P_{n-1}(t_1)=\frac{(\rho t_1)^n}{n!}\ge 0$. We now prove that $g'_n(t)\ge 0$ for $t_1\le t\le t_0$. We have
	\[g'_n(t)=\rho P_n(t_1)e^{\rho(t-t_1)}-P'_{n-1}(t)=\rho P_n(t_1)e^{\rho(t-t_1)}-\rho P_{n-2}(t)=\rho( P_n(t_1)e^{\rho(t-t_1)}- P_{n-2}(t))\]
	Thus, $g'_n(t)\ge 0$ if and only if $P_n(t_1)e^{\rho(t-t_1)}- P_{n-2}(t)\ge 0$. Repeating this argument for $n-1$ times, we end up with the condition $P_n(t_1)e^{\rho(t-t_1)}- P_{0}(t)\ge 0$, which is equivalent to $P_n(t_1)e^{\rho(t-t_1)}\ge 1$, and it holds since $P_n(t_1)\ge 1$ and $e^{\rho(t-t_1)}\ge 1$.
\end{proof}

Intuitively, Property 1 in \cref{prop: underapproximations} ensures that the curve of $P_n(t)$ always is below that of $e^{\rho t}$, Property 3 says that the under-approximation can get arbitrarily close to the exponential function, and Property 2 is a condition on the derivative of $P_n(t)$ which ensures that the resulting set is invariant. Formally, we have the following:
\begin{lemma}
	\label{lem: underapproxmiation is invariant}
        For every $n\in \NN$, the set
        \begin{align*}
          {\cal L}_n\defequals\big\{(y,t): y\ge P_n(t), 0\le t\le t_0\big\}\cup \big\{(y,t): t>t_0\big\}
        \end{align*}
        is a semi-algebraic invariant  that contains the orbit from time $0$.
\end{lemma}
\begin{proof}
	Clearly the set ${\cal L}_n$ is semi-algebraic (recall that $t_0\in \QQ$). It thus remains to prove that for every $(y_1,t_1)\in {\cal L}_n$ and for every $\delta>0$ it holds that $(e^{\rho \delta } y_1 , t_1+\delta)\in {\cal L}_n$. Denote $t=t_1+\delta$. If $t> t_0$, then the claim is trivial. Thus, assume $t_1\le t\le t_0$, and we need to prove that $P^\rho_n(t)\le e^{\rho (t-t_1)}y_1$. Since $(y_1,t_1)\in {\cal L}_n$, then $y_1\ge P^\rho_n(t_1)$, and thus for $t=t_1$ the claim holds, and Property 2 in \cref{prop: underapproximations} ensures that the inequality is maintained for all $t_1\le t\le t_0$ (by taking derivative of both sides of the inequality).
\end{proof}

\cref{prop: underapproximations,lem: underapproxmiation is invariant} provide us with an under-approximating invariant. We now turn our attention to the over-approximations. 

\begin{proposition}
  \label{prop: overapproximations}
  The over-approximations have the following properties: 
  \begin{itemize}
  \item \textbf{Property 1:} for every $\mu>1$ there exists $n_0\in \NN$ such that for all $n\ge n_0$ and $0\le t\le t_0$, we have $Q_{n,\mu}(t)\ge e^{\rho t}$,
  \item \textbf{Property 2:} for every $\mu>1$ there exists $n_0\in \NN$ such that for all $n\ge n_0$ and $0\le t_1\le t\le t_0$, we have  ${Q_{n,\mu}}'(t)\ge (Q_{n,\mu}(t_1)e^{\rho (t-t_1)})'$, 
  \item \textbf{Property 3:} for every $\epsilon>0$ there exist $\mu>1$ and $n_0\in \NN$ such that for all $n\ge n_0$, $\max_{0\le t\le t_0}\lVert Q_{n,\mu}(t)-e^{\rho t}\rVert<\epsilon $.
	\end{itemize}
\end{proposition}
\begin{proof}
  Property 3 clearly holds by Taylor's theorem and since $Q_{n,\mu}\to P_n$ uniformly as $\mu\to 1^+$. For Property 1, fix $\mu>1$, and observe that $Q_{n,\mu}(t)\to e^{\mu \rho t}$ uniformly in $[0,t_0]$ as $n\to \infty$, and since $\mu>1$, we have that $e^{\mu\rho t}\ge e^{\rho t}$.
	
  We turn to establish Property 2. Plugging the definition of $Q_{n,\mu}$ and expanding the derivatives, rewrite the property as $\rho \mu P_{n-1}(\mu t)\ge \rho P_n(\mu t_1)e^{\rho(t-t_1)}$. Cancel $\rho$, and recall from \cref{prop: underapproximations} that $P_n(t)\le e^{\rho t}$, and thus $P_n(\mu t_1)\le e^{\rho \mu t_1}$, so
	\[P_n(\mu t_1)e^{\rho(t-t_1)}\le e^{\rho \mu t_1}e^{\rho(t-t_1)}= e^{\rho((\mu-1)t_1+t)}\le e^{\rho\mu t}\]
	where the last inequality follows since $t_1\le t$.
	
	Next, from Taylor's theorem, for every $\epsilon>0$ there exists $n_0\in \NN$ such that $P_{n-1}(\mu t)\ge e^{\mu \rho t}-\epsilon$ for all $t\in [0,t_0]$. Fix $0<\epsilon< \frac{\mu -1}{\mu}$, and let $n_0$ be the corresponding threshold. By the above, it now suffices to prove that $\mu (e^{\rho\mu t}-\epsilon)\ge e^{\rho \mu t}$, which holds by our choice of $\epsilon$ for every $n\ge n_0$.
\end{proof}

We can now use \cref{prop: overapproximations} to establish the following Lemma, whose proof follows, \emph{mutatis mutandis}, the proof of \cref{lem: underapproxmiation is invariant}.

\begin{lemma}
  \label{lem: overapproxmiation is invariant}
  For every $\mu>1$ there exists $n_0\in \NN$ such that for every $n\ge n_0$, the set
  \begin{align*}
    {\cal U}_n=\big\{(y,t): y\le Q_{n\mu}(t), 0\le t\le t_0\big\}\cup \big\{(y,t): t>t_0\big\}
  \end{align*}
  is a semi-algebraic invariant that contains the orbit from time $0$. 
\end{lemma}

Combining \cref{lem: underapproxmiation is invariant,lem: overapproxmiation is invariant} and the properties in \cref{prop: underapproximations,prop: overapproximations}, we regain \cref{lem:hardness invariant}.

\lemchangebasis*
\begin{proof}
  Define $g : \CC^d\to \CC^d$ to be the injective linear map:
  \begin{align*}
    \vect v \mapsto P(\vect v P^{-1})^T, 
  \end{align*}
  and let
  \begin{align*}
    Y \defequals g(\tilde Y)\qquad\qquad \vect x_0\defequals g(\vect {\tilde x}_0)^T.
  \end{align*}

  Both $Y$ and $\vect x_0$ can be seen to be subsets of $\RR^d$ as follows. Without loss of generality, we can assume that $\pi_j(\tilde Y)=0$ for all $n+2<j\le d$, that is the projection to the last $d-n-2$ entries is zero, since $\vect {\tilde x}_0$ ignores these entries. The first $n+2$ columns of $P^{-1}$ are real numbers since they are eigenvectors that span the eigenspace corresponding to the real eigenvalues $\rho_1,\ldots,\rho_n,0$. The same is true for the first $n+2$ rows of $P$. It follows now from the definitions that $Y,\vect x_0\subset\RR^d$.  The set $Y$ is semi-algebraic because semi-algebraic sets are closed under linear maps.

  For the direct implication assume that $\tilde\inv$ is an invariant of $\tup{J,\vect {\tilde x}_0}$ with the properties in the statement. Let $\inv = g(\tilde\inv)$. We prove that $\inv$ is an invariant for $\tup{PJP^{-1},\vect x_0}$. Any point in $\inv$ can be written as
  \begin{align*}
    P(\vect x P^{-1})^T, \text{ where }\vect x\in\tilde\inv,
  \end{align*}
  hence, since $\tilde\inv$ is invariant for all $\delta\ge 0$ we have:
  \begin{align*}
    Pe^{J\delta}P^{-1}\ \cdot\ P(\vect x P^{-1})^T= P(e^{J\delta}\vect x P^{-1})^T\in \inv. 
  \end{align*}
  Moreover, by definition $\vect x_0\in\inv$ since $\vect{\tilde x}_0\in\tilde\inv$, so $\inv$ contains the whole
  orbit. The set $\inv$ can be further shown to be disjoint from $Y$, because the map $g$ is injective. The inverse
  implication follows along the same lines.

  This does not prove the lemma because $\vect x_0$ might have irrational entries. We can amend this by translating the whole system by some vector $\vect v$ such that $\vect x_0+\vect v \in \QQ^d$, which is feasible because the sets $Y+\vect v$, and $\inv+\vect v$ are semi-algebraic.
\end{proof}
